\documentclass{llncs}
\usepackage{amsmath, cite}
\usepackage{latexsym}
\usepackage{amssymb}
\usepackage{amsfonts}
\usepackage{graphicx}
\usepackage{url}
\usepackage{color}
\usepackage{authblk}
\usepackage{setspace}
\usepackage{subfigure}

\newtheorem{Definition}{Definition}
\newtheorem{Theorem}{Theorem}
\newtheorem{Lemma}{Lemma}

\newtheorem{Proposition}{Proposition}

\date{}

\begin{document}

\title{The Least-core and Nucleolus of Path Cooperative Games \thanks{The work is partially supported by National Natural Science Foundation of China (NSFC) (NO. 11271341).}}

\author{Qizhi Fang\inst{1}$^*$ \and Bo Li\inst{1}$^\dag$ \and Xiaohan Shan\inst{2}$^\ddag$ \and Xiaoming Sun\inst{2}$^\S$ }
\institute{School of Mathematical Sciences, Ocean University of China, Qingdao, China \\ $^*$\email{qfang@ouc.edu.cn}, $^\dag$\email{boli198907@gmail.com} \and Institute of Computing Technology, Chinese Academy of Sciences, Beijing, China \\ $^\ddag$\email{shanxiaohan@ict.ac.cn}, $^\S$\email{sunxiaoming@ict.ac.cn}}

\maketitle
\pagestyle{plain}
\bibliographystyle{plain}

\begin{abstract}
Cooperative games provide an appropriate framework for fair and stable profit
distribution in multiagent systems. In this paper, we study the algorithmic issues
on path cooperative games that arise from the situations where some commodity
flows through a network. In these games, a coalition of edges or vertices is
successful if it enables a path from the source to the sink in the network,
and lose otherwise. Based on dual theory of linear programming and the relationship
with flow games, we provide the characterizations on the CS-core, least-core
and nucleolus of path cooperative games. Furthermore, we show that the least-core
and nucleolus are polynomially solvable for path cooperative games defined on both
directed and undirected network.
\end{abstract}

\section{Introduction}

One of the important problems in cooperative game is how to distribute
the total profit generated by a group of agents to individual participants.
The prerequisite here is to make all the agents work together, {\it i.e.},
form a grand coalition. To achieve this goal,  the collective profit
should be distributed properly so as to minimize the incentive of
subgroups of agents to deviate and form coalitions of their own. This
intuition is formally captured by several solution concepts, such as the
core, the least-core, and the nucleolus, which will be
the focus of this paper.                                 


In this paper, we consider a kind of cooperative game models,
\emph{path cooperative games} (PC-games), arising from the situations where
some commodity (traffic, liquid or information) flows through a network.
In these games, each player controls an edge or a vertex of the network
(called edge path cooperative games or vertex edge path cooperative games,
respectively), a coalition of players wins if it enables a path from the
source to the sink, and lose otherwise. We will focus on the algorithmic
problems on game solutions of path cooperative games, especially core
related solutions.

Path cooperative games have a natural correspondence with flow games.
Flow games were first introduced by Kalai and Zemel \cite{kalai1982generalized} and
studied extensively by many researchers. When there are public arcs in the
network, the core of the flow game is nonempty if and only if there is a minimum
$(s,t)$-cut containing no public arcs. And in this case, the core can be characterized
by the minimum $(s,t)$-cuts\cite{kalai1982generalized,reijnierse1996simple}, and the nucleolus can also be
computed efficiently\cite{potters2006nucleolus,Deng:2009}. 
Recently, Aziz \emph{et al.} \cite{Aziz:2010} introduced the threshold versions of
monotone games, including PC-games as a special case. Yoram \cite{bachrach:2011}
showed that computing $\varepsilon$-core for threshold network flow games is
polynomial time solvable for unit capacity networks, and NP-hard for networks
with general capacities. For PC-games defined on series-parallel graphs, Aziz
\emph{et al.}\cite{Aziz:2010} showed that the nucleolus can be computed in
polynomial time. However, the complexity of computing the nucleolus for
general PC-games remains open, from the algorithmic point of view, the solution
concepts of general PC-games have not been systematically discussed.                


The algorithmic problems in cooperative games are especially interesting, since
except for the fairness and rationality requirements in the solution definitions,
computational complexity is suggested  be taken into consideration as another
measure of rationality for evaluating and comparing different solution concepts
(Deng and Papadimitriou~\cite{Deng:1994}). 
Until now, various interesting complexity and algorithmic results have been
investigated. On one hand, efficient algorithms have been proposed for computing
the core, the least-core and the nucleolus for, such as, assignment games
\cite{Solymosi:1994}, cardinality matching  games \cite{Kern:2003}, unit flow
games \cite{Deng:2009} and weighted voting games \cite{Elkind:2009}. On the
other hand, negative results are also given. For example, the problems of
computing the nucleolus and testing whether a given distribution belongs to the
core or the nucleolus are proved to be NP-hard for minimum spanning
tree games \cite{faigle1997complexity, faigle1998note}, flow games and linear
production games \cite{fang2002computational, Deng:2009}.   

The main contribution of this work is the efficient characterizations of the CS-core,
least-core and the nucleolus of PC-games,  based on linear programming technique
and the relationship with flow games.  These characterizations yield directly to
efficient algorithms for the related solutions. The organization of the
paper is as follows.

In section 2, the relevant definitions in cooperative game are introduced. In section
3, we first define PC-games (edge path cooperative game and vertex path
cooperative game), and then give the the characterizations of the core and CS-core.
Section 4 is dedicated to the efficient description of the least-core for PC-games.
In section 5, we prove that the nucleolus is polynomially solvable for both edge and
vertex path cooperative games.

\section{Preliminaries}\label{Preliminaries}


A cooperative game $\Gamma=(N,\gamma)$ consists of a player set
$N=\{1,2,\cdots,n\}$  and a characteristic function $\gamma: 2^N\rightarrow R$
with $\gamma(\emptyset)=0$. For each coalition $S\subseteq N$, $\gamma(S)$
represents the profit obtained by $S$ without help of other players. The set
$N$ is called the grand coalition. In what follows, we assume that
$\gamma(S)\geq 0$ for all $S\subseteq N$, and $\gamma(\emptyset) = 0$.

An imputation of $\Gamma$ is a payoff vector $x=(x_1,...x_n)$ such that
$\sum_{i\in N} x_i =\gamma(N)$ and $x_i \geq \gamma(\{i\}),\ \forall i\in N$.
The set of imputations is denoted by ${\cal I}(\Gamma)$.
Throughout this paper, we use the shorthand notation $x(S)=\sum_{i\in S}x_i$.
Given a payoff vector $x\in {\cal I}(\Gamma)$, the excess of coalition
$S\subseteq N$ with respect to $x$ is defined as: $e(x,S)=x(S)-\gamma(S)$. This value
measures the degree of $S$'s satisfaction with the payoff $x$.

\vspace{-3mm}
\subsubsection{Core.}

The \emph{core} of a game $\Gamma$, denoted by ${\cal C}(\Gamma)$,
is the set of payoff vectors satisfying that,  $x\in {\cal C}(\Gamma)$
if and only if $e(x, S)\geq 0$ for all $S\subseteq N$.
These constraints, called group rationality, ensure that no coalition
would have an incentive to split from the grand coalition N, and do better on its own.

\vspace{-3mm}
\subsubsection{Least-core.}

When ${\cal C}(\Gamma)$ is empty, it is meaningful to relax the group rationality
constraints by $e(x, S)\geq \varepsilon$ for all $S\subseteq N$. We shall find
the maximum value $\varepsilon^*$ such that the set $\{x\in {\cal I}(\Gamma):
e(x,S)\geq \varepsilon^*, \forall S\subseteq N\}$ is nonempty. This set of imputations
is called the \emph{least-core}, denoted by ${\cal LC}(\Gamma)$, and $\varepsilon^*$
is called the value of ${\cal LC}(\Gamma)$ or  ${\cal LC}$-value.

\vspace{-3mm}
\subsubsection{Nucleolus.}

Now we turn to the concept of the \emph{nucleolus}. A payoff vector
$x$ generates a $2^n$-dimensional excess vector
$\theta(x)=(e(x,S_1), \cdots,e(x,S_{2^n}))$, whose components are arranged in
a non-decreasing order. That is, $e(x,S_i)\leq e(x,S_j)$ for
$1\leq i < j \leq 2^n$. The nucleolus, denoted by $\eta(\Gamma)$, is defined
to be a payoff vector that lexicographically maximizes the excess vector
$\theta(x)$ over the set of imputations ${\cal I}(\Gamma)$. It was proved by
Schmeidler \cite{Schmeidler:1969} 
that the nucleolus of a game with the nonempty imputation set contains exactly
one element.


\vspace{-3mm}
\subsubsection{Monotone games and simple games.}

A game $\Gamma=(N,\gamma)$ is {\em monotone} if $\gamma(S')\leq \gamma(S)$
whenever $S'\subseteq S$. A game is called a simple game if it is a monotonic game
with $\gamma:2^N\rightarrow \{0,1\}$ such that $\gamma(\emptyset)=0$ and $\gamma(N)=1$.
Simple games can be usually used to model situations where there is a task to be
completed, a coalition is labeled as winning if and only if it can complete the task.
Formally, coalition $S\subseteq N$ is {\em winning} if  $\gamma(S)=1$, and {\em losing}
if $\gamma(S)=0$. A player $i$ is called a {\em veto} player if he or she belongs to
all winning coalitions. It is easy to see that, in a simple game, $i$ is a veto player
if and only if $\gamma(N)=1$ but $\gamma(N\setminus \{i\})=0$.

For simple games, Osborne\cite{Osborne:1994} and Elkind \emph{et al.} \cite{Elkind:2007}
gave the following result on the core and the nucleolus.

\begin{Lemma}\label{lem:core}
A simple game $\Gamma=(N,\gamma)$ has a nonempty core if and only if
there exists a veto player. Moreover,
\begin{enumerate}
  \item $x\in {\cal C}(\Gamma)$ if and only if $x_i=0$ for each $i\in N$ who is
  not a veto player;
  \item when ${\cal C}(\Gamma)\neq \emptyset$, the nucleolus of $\Gamma$
is given by $x_i=\frac{1}{k}$ if $i$ is a veto player and $x_i=0$ otherwise,
where $k$ is the number of veto players.
\end{enumerate}
\end{Lemma}

\vspace{-4mm}
\subsubsection{CS-core.}

Taking coalition structure into consideration, we can arrive at another solution
concept, \emph{CS-core}. Given a cooperative game $\Gamma=(N,\gamma)$,
a coalition structure over $N$ is a partition of $N$, {\it i.e.}, a collection of
subsets $\mbox{\it CS}=\{C^1,\cdots,C^k\}$ with $\cup_{j=1}^k C^j=N$  and $C^i\cap C^j=\emptyset$
for $i\neq j$ and $i,j\in \{1,\cdots,k\}$. A vector $x=(x_1,\cdots,x_n)$
is a payoff vector for a coalition structure $\mbox{\it CS}=\{C^1,\cdots,C^k\}$ if
$x_i\geq 0$ for all $i\in N$, and $x(C^j)= \gamma(C^j)$ for each $j\in\{1,\cdots,k\}$.

In general, an outcome of the game $\Gamma$ is a pair $(\mbox{\it CS},x)$,
where {\it CS} is a coalition structure and $x$ is a corresponding payoff vector.
The {\it CS}-core of the game $\Gamma=(N,\gamma)$, denoted by {${\cal C}_{cs}(\Gamma)$,
is the set of outcomes $(\mbox{\it CS},x)$ satisfying the constraints of ``group rationality".
That is,
$${\cal C}_{cs}(\Gamma)=\{(\mbox{\it CS},x):\forall C\in \mbox{\it CS}, x(C)=\gamma(C)\ \mbox{and}  \
\forall  S\subseteq N, x(S)\geq \gamma(S)\}.$$

A stronger property that is also enjoyed by many practically useful games is superadditivity.
The game $\Gamma=(N,\gamma)$ is \emph{superadditive} if it satisfies
$\gamma(S_1\cup S_2)\geq \gamma(S_1)+\gamma(S_2)$ for every pair of disjoint coalitions
$S_1,S_2\subseteq N$. This implies that the agents can earn at least as much profit by
working together within the grand coalition. Therefore, for superadditive games, it is
always assumed that the agents form the grand coalition.
For a (non-superadditive) game $\Gamma=(N,\gamma)$, we can define a new game
$\Gamma^*=(N,\gamma^*)$ by setting
$$\gamma^*(S)=\max_{{\it CS} \in {\cal CS}_{_S}}\gamma({\it CS}),
\ \forall S\subseteq N $$
where ${\cal CS}_{_S}$ denotes the space of all coalition
structures  over $S$ and $\gamma({\it CS})=\sum_{C\in {\it CS}}\gamma(C)$.
It is easy to verify that the  game $\Gamma^*$ is superadditive,
and it is called the \emph{superadditive cover} of $\Gamma$.
The relationship between the {\it CS}-core of $\Gamma$ and the core of its superadditive cover $\Gamma^*$
is presented in the following lemma \cite{greco2011complexity, chalkiadakis2011computational}.

\begin{Lemma}\label{CS-core}
A cooperative game $\Gamma=(N,\gamma)$ has nonempty {\it CS}-core if and only if its superadditive
cover $\Gamma^*=(N,\gamma^*)$ has a non-empty core. Moreover, if
${\cal C}(\Gamma^*)\neq \emptyset$, then ${\cal C}_{cs}(\Gamma)={\cal C}(\Gamma^*)$.
\end{Lemma}

\section{Path Cooperative Game and Its Core}\label{PC-game and its core}

Let $D=(V,E;s,t)$ be a connected flow network with unit arc capacity (called
unit flow network), where $V$ is the vertex set, $E$ is the arc set, $s,t\in V$
are the source and the sink of the network respectively. In this paper, an
$(s,t)$-{\em{path}} is referred to a {\em directed} path from $s$ to
$t$ that visits each vertex in $V$ at most once.

Let $U, W\subseteq V$ be a partition of the vertex set $V$ such that $s\in U$
and $t\in W$, then the set of arcs with heads in $U$ and tails in $W$ is called
an $(s,t)$-{\em edge-cut}, denoted by ${\bar E}\subseteq E$. An $(s,t)$-{\em
vertex-cut} is a vertex subset ${\bar V}\subseteq V\setminus \{s,t\}$ such that
$D\setminus {\bar V}$ is disconnected. An $(s,t)$-edge(vertex)-cut is minimum
if its cardinality is minimum. In the remainder of the paper, $(s,t)$-edge(vertex)-cuts
will be abbreviated as edge(vertex)-cut $S$ for short. Given an edge-cut ${\bar E}$,
we denote its indicator vector by ${\cal H}_{\bar E}\in \{0,1\}^{|E|}$, where
${\cal H}_{\bar E}(e)=1$ if $e\in {\bar E}$, and $0$ otherwise. The indicator
vector of a vertex-cut is defined analogously.

Now we introduce two kinds of path cooperative
games (PC-games), \emph{edge path cooperative games}
and \emph{vertex path cooperative games}.

\begin{Definition} {\bf (Path cooperative game, PC-game)} Let $D=(V,E;s,t)$ be a unit
flow network.
\begin{enumerate}
\item The associated edge path cooperative game (EPC-game)
$\Gamma_E=(E,\gamma_E)$ is:

$-$ The player set is $E$;

$-$ $\forall S\subseteq E$, $\left\{ \begin{array}{ll}
\gamma_E(S)=1 & \ \mbox{if} \ D[S]  \ \mbox{admits an }(s,t)\mbox{-path};\\
\gamma_E(S)=0 & \ \mbox{otherwise.} \end{array} \right.$

Here, $D[S]$ denotes the induced subgraph with vertex set $V$ and edge set $S$.

\vspace{3mm}
\item The associated vertex path cooperative game(VPC-game)
$\Gamma_V=(V,\gamma_V)$ is:

$-$ The player set is $V\setminus \{s,t\}$;

$-$ $\forall T\subseteq V$, $\left\{ \begin{array}{ll}
\gamma_V(T)=1 & \ \mbox{if} \ \mbox{induced subgraph } D[T]  \
\mbox{admits an }(s,t)\mbox{-path};\\
\gamma_V(T)=0 & \ \mbox{otherwise.} \end{array} \right.$
\end{enumerate}
\end{Definition}

Clearly, PC-games fall into the class of simple games.
Therefore, we can get the necessary and sufficient condition of the non-emptiness
of the core directly from Lemma \ref{lem:core}.

\begin{Proposition}\label{veto player}
Given an EPC-game $\Gamma_E$ and a VPC-game $\Gamma_V$ associated with network
$D=(V,E;s,t)$, then
\begin{enumerate}
\item ${\cal C}(\Gamma_E)\neq \emptyset$ if and only if the size of the minimum
edge-cut of $D$ is 1;
\item ${\cal C}(\Gamma_V)\neq \emptyset$ if and only if the size of the minimum
vertex-cut of $D$ is 1.
\end{enumerate}
\end{Proposition}

Moreover, when the core of a PC-game is nonempty, the only edge (vertex) in the
edge(vertex)-cut is a veto player, both the core and the nucleolus can be given directly.
In the following two sections, we only consider PC-games with empty core.

We note that PC-games also have a natural correspondence with \emph{flow games} and in what
follows, we will reveal the close relationship between flow games and PC-games.
Let $D=(V,E;s,t)$ be a unit flow network. Given $N\subseteq E$, each edge $e\in N$ is controlled
by one player,  {\it i.e.}, we can identify the set of edges $N$ with the set of players.
Edges not under control of any players, in $E\setminus N$, are called public arcs; they
can be used freely by any coalition. Thus, a unit flow network with player set $N$ is
denoted as $D\langle N\rangle=(V,E;s,t)$

\begin{Definition}{\bf (Simple flow game)} The simple flow game  $\Gamma_f\langle N\rangle=(N,\gamma)$
associated with the unit network $D\langle N\rangle$ is defined as:
\begin{enumerate}
\item The player set is $N$;
\item $\forall S\subseteq N$, $\gamma(S)$ is the value of the max-flow from $s$ to $t$
in $D[S\cup (E\setminus N)]$ (using only the edges in $S$ and public edges).
\end{enumerate}
\end{Definition}

Flow game is a classical combinatorial optimization game, which has been extensively studied.
The core of  the  flow game $\Gamma_f\langle N\rangle$ is nonempty if and only if there is a
minimum edge-cut without public edges \cite{reijnierse1996simple}.                    
In this case, the core is exactly the convex hull of the indicator vectors
of minimum edge-cuts without public edges in $D$ \cite{kalai1982generalized, reijnierse1996simple},
and the nucleolus can also be computed in polynomial time \cite{potters2006nucleolus,Deng:2009}.

Now we turn to discuss the {\it CS}-core of PC-games.
It is easy to see that for the network $D$ without public edges,
the associated flow game is the superadditive cover of the corresponding EPC-game. Thus,
the nonemptiness of {\it CS}-core of EPC-game is followed directly from Lemma \ref{CS-core}.

\begin{Proposition}\label{cs-core of EPC}
Given an EPC-game $\Gamma_E$ associated with network $D=(V,E;s,t)$, then the {\it CS}-core of $\Gamma_E$
is nonempty and it is exactly the convex hull of the indicator vectors
of minimum edge-cuts of  $D$.
\end{Proposition}

For a VPC-game, we can also establish some relationship with a flow game.
Given a network $D=(V,E;s,t)$, we transform it into a new network $D_V$ in the following way.

(1) For each $v\in V\setminus \{s,t\}$, split it into two distinct vertices $v'$ and $v''$;

(2) Connect $v'$ and $v''$ by a new directed edge $e_v=(v',v'')$. The set of all such edges is denoted by $E_V$;

(3) For original edge $e=(u,v)\in E$, transform it into a new edge $e=(u'',v')$ in $D_V$ ($s=s'=s''$ and $t=t'=t''$).

\vspace{-0.5cm}
\begin{figure}[htbp]
  \centering
  \includegraphics[width=0.50\textwidth]{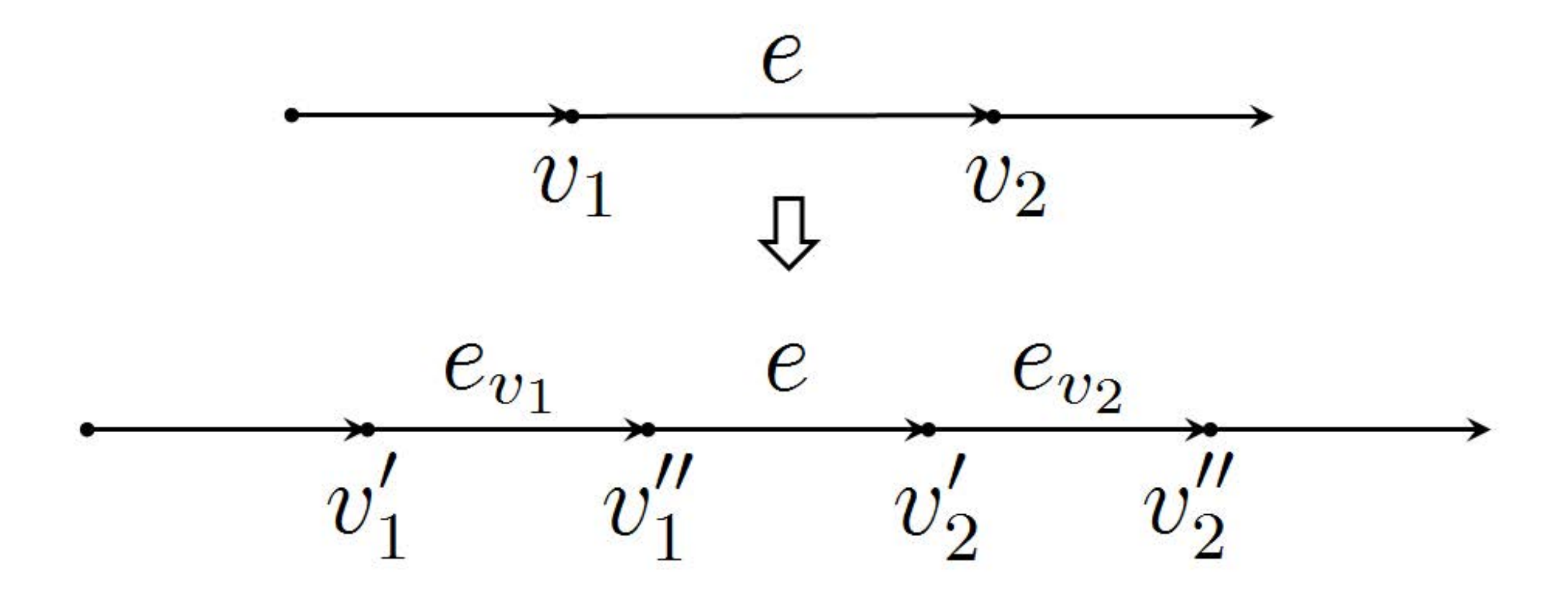}
\end{figure}
\vspace{-0.5cm}

In the new constructed network $D_V$, the player set is just the set $E_V$ and all the other
edges are viewed as public edges. It is easy to show that in the new network $D_V$, there
must be a minimum edge-cut containing only edges in $E_V$. Hence, we can verify that
the flow game associated with the network $D_V\langle E_V\rangle$ is the superadditive
cover of the corresponding VPC-game defined on $D$. Similarly,
the nonemptiness of {\it CS}-core of VPC-game is followed from Lemma \ref{CS-core} and
the results of core nonemptiness of flow games.

\begin{Proposition}\label{cs-core of VPC}
Given an VPC-game $\Gamma_V$ associated with network $D=(V,E;s,t)$, then the {\it CS}-core
of $\Gamma_V$ is nonempty and it is exactly the convex hull of the indicator vectors
of minimum vertex-cuts of $D$.
\end{Proposition}

\section{Least-core of PC-Games}\label{Least-core of PC-Game}

In this section, we first discuss the least-core of EPC-games. Throughout this
section, $\Gamma_E$ is an EPC-game associated with the network $D=(V,E;s,t)$ with
$|E|=n$. Denote by ${\cal P}$ the set of all ($s$,$t$)-path in $D$, and
$|{\cal P}|=m$. According to the definitions of EPC-game and the least-core, it
is shown that ${\cal LC}(\Gamma_E)$ can be formulated as the following linear program:

\begin{equation}\label{reduced least core}
\begin{array}{ll}
\mbox{max}& \ \varepsilon\\
\mbox{s.t.}&\left\{
\begin{array}{ll}
 x(E)=1\\
 x(P)\ge 1+\varepsilon  & \quad \forall P\in {\cal P}\\
 x_i \ge 0 & \quad \forall~i \in E
\end{array}\right.
\end{array}
\end{equation}

In spite that the number of the constrains in (\ref{reduced least core}) may be
exponential in $|E|$, the ${\cal LC}$-value and a least-core imputation
can be found efficiently by ellipsoid algorithm with a polynomial-time separation
oracle: Let $(x,\varepsilon)$ be a candidate solution for LP$({\cal LC}_E)$.
We first check whether constraints $x(E)=1$ and $x(e)\geq 0$ $(\forall e\in E)$
are satisfied. Then, checking whether $x(P)\geq 1+\varepsilon$ $(\forall P\in
{\cal P})$ are satisfied is transformed to solving the shortest $(s,t)$-path
in $D$ with respect to the edge length $x(e)$ ($\forall e\in E$), and this can
aslo be done in polynomial time.


In what follows, we aim at giving a succinct characterization of the
least-core for EPC-games. First give the linear program model of the max-flow
problem on $D$ and its dual:
\begin{equation}\label{prime}
\mbox{LP(flow):} \ \ \ \ \begin{array}{ll}
\mbox{max} & \ \sum_{j=1}^{m}{y_{j}}\\
\mbox{s.t.}&\left\{
\begin{array}{ll}
\sum\limits_{P_j:e_i\in {P_j}}y_j\le 1 &\quad i=1,2,...,n \\
y_j\ge 0 & \quad  j=1,2,...,m
\end{array}\right.
\end{array}\hspace{3mm}  {}
\end{equation}

\begin{equation}\label{dual}
\mbox{DLP(flow)}:~
\begin{array}{ll}
\mbox{min} & \ \sum_{i=1}^{n}{x_{i}}\\
\mbox{s.t.}&\left\{
\begin{array}{ll}
 \sum_{e_i:e_i\in P_j}x_i \ge 1 & \quad j=1,2,...,m \\
 x_{i}\ge 0 & \quad i=1,...,n
\end{array}\right.
\end{array}
\end{equation}

Due to {\em max-flow and min-cut} theorem, the optimum value of (\ref{prime})
and (\ref{dual}) are equal, and the set of optimal solutions of (\ref{dual})
is exactly the convex hull of the indicator vectors of the minimum edge-cut of $D$,
which is denoted by $\mathbb{C}_E$. On the other hand, it is known that the
core of the flow game $\Gamma_f$ defined on $D\langle E\rangle$ is also the
convex hull of the indicator vectors of the minimum edge-cut of $D$. Hence,
we have

\begin{Theorem}\label{least core theorem}
Let $\Gamma_E$ and $\Gamma_f$ be an EPC-game and a flow game defined on $D=(V,E;s,t)$,
respectively, $f^*$ be the value of the max-flow of $D$. Then,
$$x\in {\cal {LC}}(\Gamma_E) \ \mbox{if and only if} \ x=z/{f^*}
 \ \mbox{for some} \ z\in\mathbb{C}_E.$$
\end{Theorem}

\begin{proof}
Let $x=(1+\varepsilon)z$ be a transformation, then (\ref{reduced least core}) can
be rewritten as
\begin{equation}\label{transformation}
\begin{array}{ll}
\mbox{max}& \varepsilon \\
\mbox{s.t.}& \left\{
\begin{array}{ll}
 z(E)={1}/{(1+\varepsilon)}\\
 z(P)\ge 1  & \quad \forall P\in {\cal P}\\
 z_i \ge 0 & \quad \forall e_i \in E
\end{array}\right.
\end{array}
\end{equation}
Combining the first constraint $z(E)={1}/{(1+\varepsilon)}$ and the objective
function $\mbox{min}\{1+\varepsilon\}$, it is easy to see that linear program
(\ref{transformation}) is the same as DLP(flow) (\ref{dual}). Since the optimal
value of (\ref{dual}) is also $f^*$, Theorem \ref{least core theorem} thus
follows.\qed

\end{proof}

Based on the relationship between a VPC-game and the corresponding flow game discussed
in Section \ref{PC-game and its core}, we can obtain a similar result on the least-core for VPC-games (The proof
is omitted).

\begin{Theorem}\label{least core theorem VPC}
Let $\Gamma_V=(E,\gamma_V)$ be a VPC-game defined on $D=(V,E;s,t)$, $f^*$ be the value
of the max-flow of $D$, then
$$x\in {\cal {LC}}(\Gamma_V)\mbox{ if and only if }x={z}/{f^*} \mbox{ for some } z\in
\mathbb{C}_V.$$
Here $\mathbb{C}_V$ is the convex hull of the indicator vectors of minimum vertex-cuts
in $D$.
\end{Theorem}

Theorem \ref{least core theorem} and \ref{least core theorem VPC} show that for the
unit flow network, the least-core of the PC-game is equivalent to
the core of the corresponding flow game in the sense of scaling down by ${1}/{f^*}$.
Hence, all the following problems for PC-games can be solved efficiently:
\begin{itemize}
  \item Computing the ${\cal {LC}}$-value;
  \item Finding an imputation in ${\cal {LC}}(\Gamma_E)$ and ${\cal {LC}}(\Gamma_V)$;
  \item Checking whether a given imputation is in ${\cal {LC}}(\Gamma_E)$ or
  ${\cal {LC}}(\Gamma_V)$.
\end{itemize}

{\it Remark}. Path cooperative games have close relationship with a non-cooperative
two-person zero-sum game, called \emph{path intercept game}~\cite{washburn1995two}.
In this model, an ``evader" attempts to select a path $P$ from the source to the
sink through a given network. At the same time, an ``interdictor" attempts to select
an edge $e$ in this network to detect the evader. If the evader traverses through
arc $e$, he is detected; otherwise, he goes undetected. The interdictor aims to
find a probabilistic ``edge-inspection" strategy to maximize the average probability
of detecting the evader. While for the evader, he wants to find a "path-selection
strategy" to minimize the interdiction probability. Aziz \emph{et al.}\cite{Aziz:2011}
observed that the mixed Nash Equilibrium of path intercept games is the same as the
least-core of EPC-games. With max-min theorem in  matrix game theory, the same result
can be obtained based on the similar analysis as in the proof of Theorem \ref{least core theorem}.

\section{Nucleolus of PC-games}\label{Nucleolus of PC-game}

In this section, we aim at showing that the nucleolus of PC-games can be
computed in polynomial time. Given a game $\Gamma=(N,\gamma)$, Kopelowitz \cite{Kope:1967}
showed that the nucleolus $\eta(\Gamma)$ can be obtained by recursively solving
the following standard sequence of linear programs $SLP(\eta(\Gamma))$:
$$
\begin{array}{c}
  LP_k \\
  (k=1,2,\cdots)
\end{array}:
\begin{array}{ll}
  \max &  \varepsilon  \\
 \mbox{s.t.} & \ \left\{ \begin{array}{ll}
 x(S)=\gamma(S)+ \varepsilon_r, & \  \forall S\in {\cal J}_r \quad r=0,1,\cdots, k-1 \\
 x(S)\geq \gamma(S)+\varepsilon_r,  & \  \forall \emptyset\neq S\subset N \setminus  \cup_{r=0}^{k-1}{\cal J}_r \\
 x\in {\cal I}(\Gamma).
\end{array} \right.
\end{array}$$
Initially, set ${\cal J}_0=\{\emptyset, N\}$ and $\varepsilon_0=0$.
The number $\varepsilon_r$ is the optimal value of the $r$-th program $LP_r$,
and $J_r=\{S\subseteq N: x(S)=\gamma(S)+\varepsilon_r,\forall x\in X_r\}$,
where $X_r=\{x\in R^n:(x,\varepsilon_r)$ is an optimal solution of $LP_r\}$.

As in the last section, we first discuss the nucleolus of EPC-games. Let $\Gamma_E$
be the  EPC-game associated with network $D=(V,E;s,t)$ with $|E|=n$, ${\cal P}$ be
the set of all $(s,t)$-paths and $f^*$ be the value of the max-flow of $D$. Denote
${\cal E}_{\Gamma}$ be the set of coalitions consisting of one-edge coalitions and
path coalitions, {\it i.e.},
$${\cal E}_{\Gamma}=\{\{e\}:e\in E\}\cup \{P\subseteq E: P \in {\cal P}
\ \mbox{is an} \ (s,t)\mbox{-path}\}.$$

We show that the sequential linear programs $SLP(\eta(\Gamma_E))$ of EPC-game $\Gamma_E$
can be simplified as follows.
\begin{equation}\label{LP_k':transformation}
LP_k':~
\begin{array}{ll}
\mbox{max}& \quad \varepsilon\\
\vspace{1mm} \mbox{s.t.}&\left\{
\begin{array}{ll}
 x(e)=\varepsilon_r, & \forall e\in E_r,r=0,1,...,k-1\\
 x(e) \geq \varepsilon, & \forall e\in {E\backslash \bigcup_{r=0}^{k-1}E_r}\\
 x(P)={1}/{f^*} + \varepsilon_r,  &  \forall P\in {\cal P}_r,r=0,1,...,k-1\\
 x(P)\geq {1}/{f^*} + \varepsilon, &  \forall P\in {\cal P}\backslash
 \bigcup_{r=0}^{k-1}{{\cal P}_r}\\
 x(e)\geq 0, & \forall e\in E\\
 x(E)=1.
\end{array}\right.
\end{array}
\end{equation}
where $\varepsilon_r$ is the optimum value of $LP_r$,
$X_r=\{x\in R^n:(x,\varepsilon_r)$ is an optimal solution of $LP_r\}$,
${\cal P}_r=\{P\in {\cal P}: x(P)=1+\varepsilon_r,\forall x\in X_r\}$
and $E_r=\{e\in E: x(e)=\varepsilon_r, \forall x\in X_r \}$.
Initially, $\varepsilon_0=0$, ${\cal P}_0=\emptyset$ and $E_0=\emptyset$.

\begin{Proposition}\label{transformation proposition}
The nucleolus $\eta(\Gamma_E)$ of EPC-game $\Gamma_E$ defined on the network
$D=(V,E;s,t)$ can be obtained  by computing the linear programs $LP_k'$ in
\emph{(\ref{LP_k':transformation})}.
\end{Proposition}

\begin{proof}
Firstly, we show that in sequential linear programs $SLP(\eta(\Gamma))$,
only the constrains corresponding to the the coalitions in ${\cal E}_{\Gamma}$
({\it i.e.}, the one-edge coalitions and path coalitions) are necessary in
determining the nucleolus $\eta(\Gamma_E)$.

In fact, for any winning coalition $S\subseteq N$ (not a path), $S$ can be
decomposed into a path $P$ and some edges $E'=S\backslash E(P)$. Then,
$$x(S)-\gamma(S)=x(P)-1+\sum_{e\in E'}{x(e)}\geq x(P)-1.$$
Since $x(e)\geq 0$ for all $e\in E'$, $S$ cannot be fixed before $P$ or any
$e\in E'$. After $P$ and all $e\in E'$ are fixed, $S$ is also fixed, {\it i.e.},
$S$ is redundant. If $S$ is a losing coalition, then $S$ is a set of edges with
$\gamma(S)=0$ and $x(S)-\gamma(S)=\sum_{e\in S}{x(e)}\geq x(e), \forall e\in S$.
That is to say, $S$ cannot be fixed before any $e\in S$.
When all edges in  $S$ are fixed, $S$ is fixed accordingly, \emph{i.e.} $S$ is
also redundant in this case. Therefore, deleting all the constrains corresponding
to the coalitions not in ${\cal E}_{\Gamma}$  will not change the result of
$SLP(\eta(\Gamma))$.

The key point in remainder of the proof is the correctness of the third and the
forth constrains in (\ref{LP_k':transformation}), where we replace the original
constrains $x(P)=1+\varepsilon_r$ and $x(P)\geq 1 + \varepsilon$ in $SLP(\eta(\Gamma))$
with new constrains $x(P)={1}/{f^*} +\varepsilon_r$ and $x(P)\geq {1}/{f^*}+\varepsilon$,
respectively.

In the process of solving the sequential linear programs,
the optimal values increase with $k$. Since ${\cal C}(\Gamma_E)=\emptyset$,
we know $\varepsilon_1<0$. Note that we can always find an optimal solution
such that $\varepsilon_1>-1$ (for example $x(e)=\frac{1}{n},\forall e\in E$
is a feasible solution of the linear programming of ${\cal LC}(\Gamma_E)$).

We can divide the process into two stages. The first stage is the programs with
$-1<\varepsilon_r < 0$. In this case, the constraints $x(e)\geq \varepsilon,
\forall e\in E$ cannot effect the optimal solutions of the current programs,
because $x(e)\geq 0$. Ignoring the invalid constraints we can get
(\ref{LP_k':transformation}) directly.

The second stage is the programs with $\varepsilon_r \geq 0$. When the programs
arrive at this stage, we can claim that all paths have been fixed. Otherwise, if
there is a path satisfying $x(p)=1+\varepsilon_r\geq 1$, then we have $x(p)=1$
(note $x(E)=1$), contradicting with the precondition the value of maximum flow
$f^*\geq 2$. Then we can omit the path constraints in this stage and then this
implies (\ref{LP_k':transformation}).

This completes the proof of Proposition \ref{transformation proposition}.
\qed
\end{proof}

In the following, by making us the known results on the nucleolus of flow
games, we shall show that the nucleolus of PC-games can be solved
in polynomial time.  Let $\Gamma_f=(E,\gamma)$ be the flow game
defined on the unit flow network $D=(V,E;s,t)$. It is easy to
show that the sequential linear programs $LP(\eta(\Gamma_f))$ can be simplified
as $\widetilde{LP}_k,(k=1,2,...)$ :
\begin{equation}
\widetilde{LP}_k:~
\begin{array}{ll}
\mbox{max}& \quad \varepsilon\\
\vspace{1mm} \mbox{s.t.}&\left\{
\begin{array}{ll}
 x(e)=\varepsilon_r & \quad \forall e\in E_r,r=0,1,...,k-1\\
 x(P)=1 + \varepsilon_r  & \quad \forall P\in {\cal P}_r,r=0,1,...,k-1\\
 x(e) \geq \varepsilon & \quad \forall e\in {E\backslash \bigcup_{r=0}^{k-1}E_r}\\
 x(P)\geq 1 + \varepsilon & \quad \forall P\in {\cal P}\backslash
 \bigcup_{r=0}^{k-1}{{\cal P}_r}\\
 x(E)=f^*,\\
\end{array}\right.
\end{array}
\end{equation}
where $\varepsilon_r$ is the optimum value of $\widetilde{LP}_r$,
$X_r=\{x\in R^n:(x,\varepsilon_r)$ is an optimal solution of $\widetilde{LP}_r\}$,
${\cal P}_r=\{P\in {\cal P}: x(P)=1+\varepsilon_r,\forall x\in X_r\}$ and
$E_r=\{e\in E: x(e)=\varepsilon_r, \forall x\in X_r \}$. Initially,
$\varepsilon_0=0$, ${\cal P}_0=\emptyset$ and $E_0=\emptyset$.

Deng \emph{et al.} \cite{Deng:2009} proved that the sequential linear programs
$\widetilde{LP}_k,(k=1,2...)$ can be transformed to another sequential linear programs
with only polynomial number of constrains, and it follows that the nucleolus of flow
game $\eta(\Gamma_f)$ can be found efficiently. Futhermore, Potters {\it et al.} \cite{potters2006nucleolus},      
show that the nucleolus of flow games with public edge can also be found in polynomial
time when the core is nonempty. Based on these known results, we discuss the   
algorithmic problem on the nucleolus of PC-games in the following theorems.

\begin{Theorem}\label{nucleolus of EPC}
Let $\Gamma_E$ and $\Gamma_f$ be the EPC-game and flow game defined on a unit flow
network $D=(V,E;s,t)$, respectively. The nucleolus of $\Gamma_E$ can be computed
in polynomial time. Furthermore,
$$x\in\eta(\Gamma_E)\ \mbox{\it if and only  if} \ z = x\cdot{f^*}\in \eta(\Gamma_f),$$
where $f^*$ is the value of the max-flow of $D$.
\end{Theorem}

\begin{proof}
Notice that the dimension of the feasible regions of $LP_k' (k=1,2...)$ \label{LP_k':transformation}
decreases in each step, so we can end up the process within at most $|N|$ steps.

The key point here is to show that there is a one-to-one correspondence between the
optimal solutions of $\widetilde{LP}_k$ (6) and that of $LP_k'$ (5) ($\forall k=1,2,\cdots$).

We first prove that if $({z}^*, \tilde{\varepsilon}^*)$ is an optimal solution of
$ \widetilde{LP}_k$ (6), the  $(x^*,\varepsilon^*)= ({z}^*/f^*, \tilde{\varepsilon}^*/f^*)$
is an optimal solution of  $LP_k'$ (5).

When $k=1$, we have  $E_0=\emptyset$, ${\cal P}_0=\emptyset$ in $LP_1'$. And it is
easy to check the feasibility and the optimality of $(z^*,\varepsilon^*)$ in $LP_1'$. To
continue the proof recursively, we need to explain $E_1=\tilde{E}_1$ and
${\cal P}_1= \tilde{\cal P}_1$, {\it i.e.}, the constrains which become tight in every iteration
are exactly the same in the two linear programs. For each $e\in E$,
if ${z}^*(e)=\tilde{\varepsilon}^*$, then
${{x}^*(e)}={{z}^{*}}(e)/{{{f}^{*}}}={{{\tilde{\varepsilon} }^{*}}}/{{{f}^{*}}}=\varepsilon^* $.
And if $z^*(e)>\tilde{\varepsilon}^*$, then we have $x^*(e)>\varepsilon^*$. Thus, $E_1=\tilde{E}_1$.
${\cal P}_1= \tilde{\cal P}_1$ can be shown analogously.
The other direction of the result can be shown similarly. That is, the conclusion holds for $k=1$.

For the rest iteration $k=2,3,\cdots$, the proof can be carried out in a same way. Here we omit
the detail of the proof. Since the nucleolus of flow game can be found in polynomial time, it
follows that the nucleolus of EPC-game is also efficiently solvable.
\qed
\end{proof}

As for the nucleolus of VPC-games, we also show that it is polynomially solvable
based on the relationship between a VPC-game and the corresponding flow game
demonstrated in Section 3. Due to the space limitation, the proof of the following theroem is omitted.


\begin{Theorem}\label{nucleolus of VPC}
The nucleolus of VPC-games can be solved in polynomial time.
\end{Theorem}

\subsubsection{PC-games on undirected networks.}

Given an undirected network $D=(V,E;s,t)$, we construct a directed network
$\overrightarrow{D}=(V,\overrightarrow{E};s,t)$ derived from $D$ as follows (see the following figure):
\begin{enumerate}
\item For edge $e\in E$ with end vertices $v_1$ and $v_2$,
transform it into two directed edges $\overrightarrow{e}_{v_1}=(v_{11},v_{12})$
and $\overrightarrow{e}_{v_2}=(v_{21},v_{22})$;
\item Connect the two directed edges into a directed cycle via two supplemental
directed edges $\overrightarrow{e}_1$ and $\overrightarrow{e}_2$.
\end{enumerate}
\vspace{-0.5cm}
\begin{figure}[htbp]
  \centering
  \includegraphics[width=0.50\textwidth]{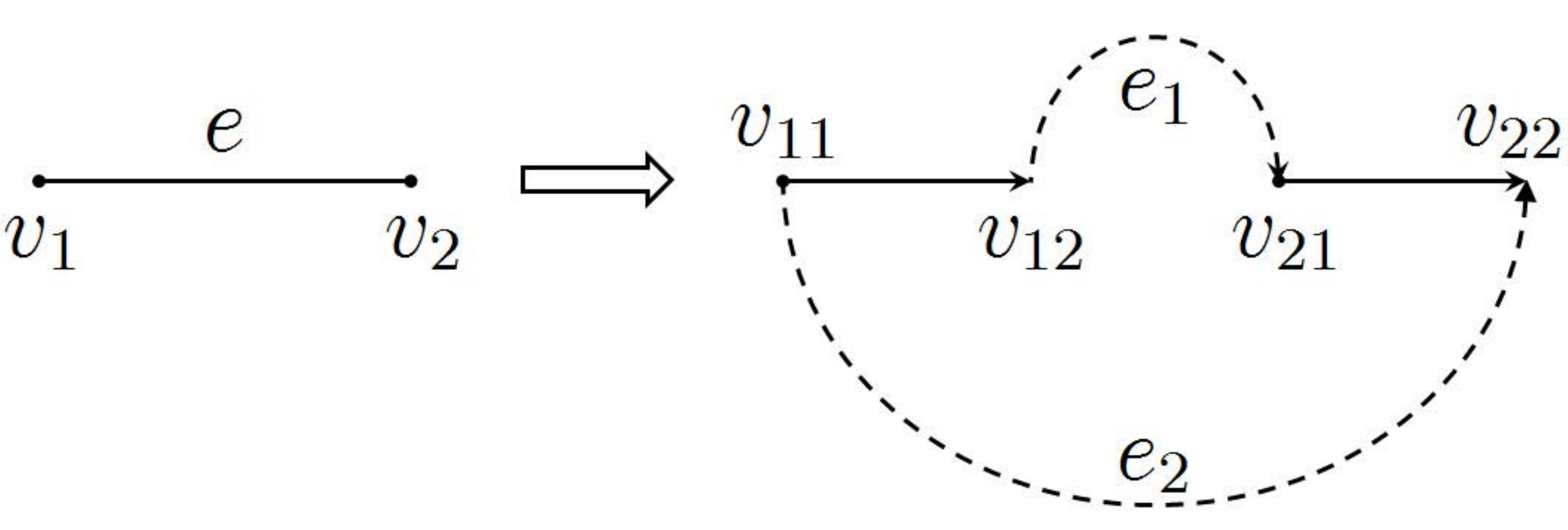}
\end{figure}
\vspace{-0.5cm}
Thus, the EPC-game defined on undirected network $D=(V,E;s,t)$ is transformed to
an  EPC-game defined on the constructed directed network $\overrightarrow{D}=(V,\overrightarrow{E};s,t)$.
Furthermore, it is easy to check that there exists one-to-one correspondence  for the
game solution (such as, the core, the least-core and the nucleolus) between the two games.
As for a VPC-game defined on an undirected network, we  first transform it into
 EPC-game  on an undirected network as demonstrated in Section 3, and then transform
it to EPC-game on a directed network in the same way as above. Henceforth, the
algorithmic results for PC-games can be generalized from directed networks to undirected networks.

\begin{Theorem}\label{nucleolus of undirected}
Computing the least-core and the nucleolus can be done in polynomial time for both
EPC-games and VPC-games defined on undirected networks.
\end{Theorem}

\bibliography{ref}

\begin{thebibliography}{10}

\bibitem{Aziz:2010}
Haris Aziz, Felix Brandt, and Paul Harrenstein.
\newblock Monotone cooperative games and their threshold versions.
\newblock In {\em Proceedings of the 9th International Conference on Autonomous
  Agents and Multiagent Systems}, volume~1, pages 1107--1114, 2010.

\bibitem{Aziz:2011}
Haris Aziz and Troels~Bjerre S{\o}rensen.
\newblock Path coalitional games.
\newblock {\em arXiv preprint arXiv:1103.3310}, 2011.

\bibitem{bachrach:2011}
Yoram Bachrach.
\newblock The least-core of threshold network flow games.
\newblock In {\em Mathematical Foundations of Computer Science 2011}, pages
  36--47. Springer, 2011.

\bibitem{chalkiadakis2011computational}
Georgios Chalkiadakis, Edith Elkind, and Michael Wooldridge.
\newblock Computational aspects of cooperative game theory.
\newblock {\em Synthesis Lectures on Artificial Intelligence and Machine
  Learning}, 5(6):1--168, 2011.

\bibitem{Deng:2009}
Xiaotie Deng, Qizhi Fang, and Xiaoxun Sun.
\newblock Finding nucleolus of flow game.
\newblock {\em Journal of combinatorial optimization}, 18(1):64--86, 2009.

\bibitem{Deng:1994}
Xiaotie Deng and Christos~H Papadimitriou.
\newblock On the complexity of cooperative solution concepts.
\newblock {\em Mathematics of Operations Research}, 19(2):257--266, 1994.

\bibitem{Elkind:2007}
Edith Elkind, Leslie~Ann Goldberg, Paul~W Goldberg, and Michael Wooldridge.
\newblock Computational complexity of weighted threshold games.
\newblock In {\em Proceedings of the National Conference on Artificial
  Intelligence}, volume~22, page 718, 2007.

\bibitem{Elkind:2009}
Edith Elkind and Dmitrii Pasechnik.
\newblock Computing the nucleolus of weighted voting games.
\newblock In {\em Proceedings of the 12th Annual ACM-SIAM Symposium on Discrete
  Algorithms}, pages 327--335, 2009.

\bibitem{faigle1997complexity}
Ulrich Faigle, Walter Kern, S{\'a}ndor~P Fekete, and Winfried Hochst{\"a}ttler.
\newblock On the complexity of testing membership in the core of min-cost
  spanning tree games.
\newblock {\em International Journal of Game Theory}, 26(3):361--366, 1997.

\bibitem{faigle1998note}
Ulrich Faigle, Walter Kern, and Jeroen Kuipers.
\newblock Note computing the nucleolus of min-cost spanning tree games is
  np-hard.
\newblock {\em International Journal of Game Theory}, 27(3):443--450, 1998.

\bibitem{fang2002computational}
Qizhi Fang, Shanfeng Zhu, Maocheng Cai, and Xiaotie Deng.
\newblock On computational complexity of membership test in flow games and
  linear production games.
\newblock {\em International Journal of Game Theory}, 31(1):39--45, 2002.

\bibitem{greco2011complexity}
Gianluigi Greco, Enrico Malizia, Luigi Palopoli, and Francesco Scarcello.
\newblock On the complexity of the core over coalition structures.
\newblock In {\em IJCAI}, volume~11, pages 216--221. Citeseer, 2011.

\bibitem{kalai1982generalized}
Ehud Kalai and Eitan Zemel.
\newblock Generalized network problems yielding totally balanced games.
\newblock {\em Operations Research}, 30(5):998--1008, 1982.

\bibitem{Kern:2003}
Walter Kern and Dani{\"e}l Paulusma.
\newblock Matching games: the least-core and the nucleolus.
\newblock {\em Mathematics of Operations Research}, 28(2):294--308, 2003.

\bibitem{Kope:1967}
Alexander Kopelowitz.
\newblock Computation of the kernels of simple games and the nucleolus of
  n-person games.
\newblock Technical report, DTIC Document, 1967.

\bibitem{Osborne:1994}
Martin~J Osborne and Ariel Rubinstein.
\newblock A course in game theory.
\newblock {\em Cambridge, Massachusetts}, 1994.

\bibitem{potters2006nucleolus}
Jos Potters, Hans Reijnierse, and Amit Biswas.
\newblock The nucleolus of balanced simple flow networks.
\newblock {\em Games and Economic Behavior}, 54(1):205--225, 2006.

\bibitem{reijnierse1996simple}
Hans Reijnierse, Michael Maschler, Jos Potters, and Stef Tijs.
\newblock Simple flow games.
\newblock {\em Games and Economic Behavior}, 16(2):238--260, 1996.

\bibitem{Schmeidler:1969}
David Schmeidler.
\newblock The nucleolus of a characteristic function game.
\newblock {\em SIAM Journal on applied mathematics}, 17(6):1163--1170, 1969.

\bibitem{Solymosi:1994}
Tam{\'a}s Solymosi and Tirukkannamangai~ES Raghavan.
\newblock An algorithm for finding the nucleolus of assignment games.
\newblock {\em International Journal of Game Theory}, 23(2):119--143, 1994.

\bibitem{washburn1995two}
Alan Washburn and Kevin Wood.
\newblock Two-person zero-sum games for network interdiction.
\newblock {\em Operations Research}, 43(2):243--251, 1995.

\end{thebibliography}

\end{document}